\documentclass[twocolumn]{autart}
\usepackage{mathrsfs}
\usepackage[english]{babel}

\usepackage[noadjust]{cite}
\usepackage{amsmath,amsfonts,amssymb,graphicx}
\usepackage{algorithm}
\usepackage{pdfsync}
\usepackage{booktabs}
\usepackage{algpseudocode}
\usepackage{graphics}
\usepackage{epsfig}
\usepackage{multirow}
\usepackage{subfigure}
\usepackage{tikz}
\usepackage{bm}
\usepackage{flushend}

\usepackage{mathrsfs}
\usepackage{color}
\usepackage{tikz-cd}
\usepackage{booktabs}
\usepackage{arydshln}
\usepackage{adjustbox}

\newtheorem{theorem}{Theorem}[section]
\newtheorem{lemma}[theorem]{Lemma}

\newtheorem{proposition}[theorem]{Proposition}
\newtheorem{remark}[theorem]{Remark}

\newcommand{\R}{{\rm  I\kern-2pt R}}
\renewcommand{\Re}{{\rm  I\kern-2pt R}}

\begin{document}

\begin{frontmatter}

\title{Fusion Estimation in Multi-sensor Systems for Data Packets with Disrupted Identities}

\author[label1]{Haoyuan Xu}\ead{haoyuanxu2021@qq.com},
\author[label1]{Yuzhe Li}\ead{yuzheli@mail.neu.edu.cn},

\address[label1]{The State Key Laboratory of Synthetical Automation for Process Industries, Northeastern University, Shenyang 110819, China.}

\maketitle

\begin{abstract}
	
	In this paper, we explore the problem of fusion estimation for a multi-sensor system where the identity of the data packet received by each sensor may be disrupted or incorrect due to confusion in device identity allocation, communication protocol defects, or the lack of a clear sensor identifier.
	This can result in a random shuffle of the data components during the fusion estimation process, compromising the performance of the fusion estimation. To address this issue, we introduce the concepts of permutations and symmetry groups to describe this phenomenon as data packet permutation. We construct statistics to simplify the information set, developing two algorithms: a Bayesian approach, which performs fusion using posterior arrangement probabilities, and a greedy approach, which effectively improves estimation performance by guessing the likely data arrangement. We compare these two algorithms and demonstrate that both are expectation error-bounded. 
	We improve algorithms for information-scarce scenarios. By employing the expectation-maximization algorithm, we fill in the prior information of data arrangement where the correct convergence is proven. Finally, we present numerical simulations to validate our results.
	
\end{abstract}


\end{frontmatter}

\section{Introduction} \label{Section:Introduction}

	In the era of rapid technological advancement, the integration of multiple sensors into cohesive systems has emerged as a pivotal strategy in enhancing the capabilities of collecting information from the environment. Multi-sensor systems have found crucial applications across diverse domains, including robotics, autonomous vehicles, environmental monitoring, and healthcare~\cite{muzammal2020multi}, underscoring their widespread importance and applicability. These systems leverage the unique strengths of each sensor type, such as cameras, LiDAR, and thermal sensors~\cite{kong2020multi}, to overcome the limitations inherent in single-sensor setups. 	
	{
	By fusing data from these varied sources, the system gains a more comprehensive understanding of the environment, achieving higher accuracy, reliability, and robustness in perception and decision-making. Hence, fusion estimation has emerged as a key technique to mitigate sensor noise and bias, enabling accurate and robust state estimation in distributed networks~\cite{wang2017convergence}. The increasing reliance on multi-sensor communication also raises significant security concerns.}

	In numerous practical scenarios, the accurate parsing of the source of data packets by the receiving end can be compromised by a variety of factors. For example, the incorporation of additional sensors and alterations in network topology may necessitate the reassignment of sensor identity documents (IDs) within the network. The reassignment typically relies on ID management~\cite{wang2010id} algorithms, as discussed in~\cite{xu2020identity}. While such types of approaches can improve the robustness of the system by decentralization or dynamic revocation, they may also introduce communication overhead and require additional trust dependencies.
	The loss of data packet identity is a common phenomenon with diverse causes. Temporal misalignment from delays, jitter, or asynchronous clocks can hinder reliable source association~\cite{nila2024synchronization}. In privacy-preserving applications, packets are deliberately anonymized to protect confidentiality, forcing the center to process data without origin~\cite{aravind2023tracing}. Protocol issues, such as dynamic host configuration protocol (DHCP) expiration or address collisions, also yield ambiguous sources in large networks~\cite{aldaoud2021dhcp}. Beyond payload errors, header corruption may irreversibly damage identifiers, leaving usable but anonymous data~\cite{hermans2014all}. In resource-constrained systems like wireless sensor networks (WSNs) and internet of things (IoT) devices, metadata may be omitted to reduce overhead, producing anonymous packets at the fusion center~\cite{hudda2025review}.
	As a consequence of the fusion estimation process, such a phenomenon can be conceptualized as a reshuffle of data packets collected from multiple sensors into an unknown arrangement. These defective data packets are usually discarded and wasted. 
	Intuitively, a more reasonable data fusion mechanism may utilize such defective data packets to improve estimation performance.

	The disruption of the identity of the data packet causes the system to be unable to distinguish packets among different data sources, essentially breaking the correct order in which data packets should be processed. Existing works have attempted to address the data fusion problem under potentially incorrect packet sequencing. For example, the optimal filtering with out-of-sequence measurements~\cite{zhang2005optimal} is investigated in a system where sensors are arranged orderly but with chaotic data packet time. A Bayesian processing method for out-of-sequence measurements only due to different time delays is derived in~\cite{garcia2021continuous} for dealing with the multi-target tracking problem in a continuous-time system. An important class of methods is the Probabilistic Data Association Filter (PDAF)~\cite{bar1975tracking}, which originated from target tracking research. Subsequent developments include the Joint PDAF~\cite{rezatofighi2015joint}. A recent study~\cite{zhao2025state} introduced several heuristic algorithms based on comparing bias magnitudes in Kalman filtering to reconstruct packet identities; however, these methods cannot be proven optimal in terms of reconstruction probability or fusion error. Although the aforementioned studies exhibit a degree of efficacy in data retrieval, their exploitation of data patterns in addressing the issue of data rearrangement caused by identity destruction remains relatively superficial.

	Motivated by state estimation theory and statistical learning theory, this paper considers the utilization of statistical characteristics of data to re-estimate the disrupted identity of data packets in a multi-sensor system. 
	
	The main contributions of this paper are summarized as follows:
	\begin{enumerate}
		\item
		{
		By employing concepts from permutation and symmetric groups, we reconstruct the data packets into an auxiliary vector, which is characterized by distinct explicit Gaussian distributions corresponding to different data arrangements.} {This statistic admits stable distribution parameters even when the underlying state is unstable, making it an efficient and analytical prior for inferring packet permutations.}

		\item 
		Two data fusion approaches are developed to overcome the challenge of randomly arranged data packets and prevent the increasing computational complexity associated with posterior estimates over time. The first approach, Bayesian fusion, performs approximately optimal weighted fusion using posterior arrangement probabilities; the second approach, a greedy data fusion method, recovers the data packets in the maximum-likelihood arrangement. Both approaches significantly improve state estimation performance and derive the average estimation error.
		
		\item 
		We also improve corresponding algorithms for scenarios lacking information. {For the scenario without the \textit{a priori} probability distribution of arrangements, we embed the Expectation-Maximization (EM) algorithm into the Bayesian approach, learning the prior information to assist the fusion estimation.} The correct convergence of the EM algorithm is proven when the data is sufficiently abundant, enabling them to achieve performance with a known \textit{a priori} probability distribution with the increase of data collection. 
		
	\end{enumerate}
	
	The remainder of this paper is organized as follows. Section~\ref{Section:Problem Formulation} establishes the multi-sensor system model and presents the problem statement. Section~\ref{section:The Bayesian Approach with P} develops the fusion estimation algorithms to deal with the randomly shuffled data packets with known probabilities of arrangements. Section~\ref{section:Learning Prior Information with the EM Algorithm} proposes the fusion estimation algorithms for the scenario without probabilities of shuffling order. 
	Simulation results are shown in Section~\ref{section:simulation}. Finally, Section~\ref{Section:Conclusion} concludes this paper.

	Notations:~$\mathbb{N}$ and~$\mathbb{R}$ are the sets of natural and real numbers, respectively.~$\mathbb{N}_{+}$ denotes the set of positive integers.~$\mathbb{R}^{n}$ denotes the~$n$-dimensional Euclidean space. Denote~$ \otimes $ as the Kronecker product.~$ \mathrm{Tr}\{\cdot\} $ denotes the trace operator.~$X^\mathrm{T}$ represents the transpose matrix of~$X$.~$ X \geq 0 $ (or~$ X > 0 $) represents that~$ X $ is positive semi-definite (or positive definite).~$ X \geq Y $ (or~$ X > Y $) represents~$ X - Y \geq 0 $ (or~$ X - Y > 0 $). The~$ n \times n $ identity matrix is denoted as~$I_n$, also written as~$ I $ without ambiguity. The matrix composed of identity matrices is denoted as~$I_{n,N}^\mathrm{T} = [ I_n, \cdots, I_n ]$ with~$N$ matrices~$I_n$.~
	$\mathbb{E}[\cdot]$ denotes the expectation of a random variable.~$ \lambda_i (X) $ represents the~$i$th largest eigenvalue of~$ X \in \mathbb{S}^{n}_{+} $.
	Let~$ \| x \| $ denote the~$ L^2 $ norm of a vector~$ x \in \mathbb{R}^n $. The Frobenius norm of matrices is represented as~$ \left\| \cdot \right\|_{\mathrm{F}} $. 
	To simplify the symmetrical notation,~$ ( X ) Y ( \cdot )^\mathrm{T} $ means~$ ( X ) Y ( X )^\mathrm{T} $, and~$ X +( \cdot )^\mathrm{T} $ means~$ X + X^\mathrm{T} $. {We denote~$p(x)$ as the probability density function of the random vector~$x$.} {The Gaussian probability density function of the random vector~$x$ with the mean~$ \mu $ and the covariance matrix~$ \Sigma $ is denoted as~$\mathcal{N}_x \left( \mu, \Sigma \right)$.}

\section{Problem Formulation}  \label{Section:Problem Formulation}

\subsection{System Model}

We consider a state-space model of a physical process with stochastic noises:
\begin{align} 
	x_{k+1} 	& =  A 		x_{k} + w_k, \label{eqn:process1}
	\\
	y_{k,i} 	& =  C_i	x_{k} + v_{k,i},	\label{eqn:process2}
\end{align}
where~$ x_{k} \in \mathbb{R}^{n} $ is the state vector, and~$ y_{k,i} \in \mathbb{R}^{m_i} $ is the measurement collected by smart sensor~$ i \in \{1,\ldots,N\} $. The process noise~$ w_k \sim \mathcal{N} (0, Q) $ with~$ Q \geq 0 $ and measurement noises~$ v_{k,i} \sim \mathcal{N} (0, R_i) $ with~$ R_i > 0 $ are independent white Gaussian noises, respectively, where~$ \mathbb{E}[ w_i v_j^\mathrm{T}] = 0 $. The initial state~$ x_0 \sim \mathcal{N} (0, \Pi) $ where~$ \Pi \geq 0 $, and~$ x_0 $ is independent of~$w_k$ and~$v_k$. Besides, the pairs~$(A, Q^{\frac{1}{2}})$ and~$(A, C_i)$ are controllable and observable, respectively. Each smart sensor has computation capability to generate a local state estimation in a time-step and transmit it to the fusion center via networks. 

Processing data packets with false identities during information fusion may lead to bias in state estimation results, and this bias may propagate over time. We consider the state estimation algorithm as a two-layer fusion structure \cite{sun2004multi}---local estimate and fusion estimate. The primary benefit of this architecture is that the local state estimation remains uncontaminated, as it does not incorporate external data. Consequently, the local state estimation consistently provides a reliable backup for the estimation results. 
For the local estimate, each sensor runs a local Kalman filter:
\begin{align} 
	\hat{x}^{-}_{k+1,i} &= A \hat{x}_{k,i},												\label{eqn:standard x-}
	\\
	P^{-}_{k+1,i} 		&= A P_{k,i} A^{\mathrm{T}} + Q, 								\label{eqn:standard P-}
	\\	
	\hat{x}_{k,i} 	&= \hat{x}^{-}_{k,i} + K_{k,i} (y_{k,i} - C \hat{x}_{k,i}^{-}),									\label{eqn:standard x}
	\\	
	K_{k,i} 			&= P^{-}_{k,i} C_i^{\mathrm{T}} ( C_i P^{-}_{k,i} C_i^{\mathrm{T}} + R_i)^{-1},	\label{eqn:standard K}
	\\	
	P_{k,i} 			&= (I - K_{k,i} C_i) P^{-}_{k,i},										\label{eqn:standard P}
\end{align}
where~$ \hat{x}^{-}_{k,i} $ and~$ \hat{x}_{k,i} $ are the \textit{a priori} and the \textit{a posteriori} minimum mean square error (MMSE) estimates of~$x_k$ at system time~$ k $;~$ P^{-}_{k,i} $ and~$ P_{k,i} $ are the corresponding estimation error covariance matrices.~$ \hat{x}_{k,1},\ldots, \hat{x}_{k,N} $ are transmitted to the fusion center via networks. 
The fusion center integrates the received state estimation results by the linear minimum variance criterion proposed in \cite{kim1994development}, where it has private side information of measurements, denoted as~$y_{k,0}$, generating a private local state estimation~$\hat{x}_{k,0}$ by the Kalman filter algorithm. If the network is completely reliable, the fusion algorithm follows 
\begin{align} \label{eqn:fusion x}
	\bar{x}_k =\sum_{i=0}^N{W_{k,i}}\hat{x}_{k,i}.
\end{align}
The matrix weights~$ W = [W_{k,0},\ldots,W_{k,N}] $, the corresponding estimation error~$\bar{P}_k = \mathbb{E}[(x_k - \bar{x}_k)(x_k - \bar{x}_k)^\mathrm{T} | ]$, and the local estimation error covariance matrix
\begin{align}
	\mathcal{P} _k =\left[ \begin{matrix}
		P_{k,00}&		\cdots&		P_{k,0N}\\
		\vdots&		\ddots&		\vdots\\
		P_{k,N0}&		\cdots&		P_{k,NN}\\
	\end{matrix} \right],
\end{align}
with~$ P_{k,ij} = \mathbb{E}[(x_k-\hat{x}_{k,i})(x_k-\hat{x}_{k,j})^\mathrm{T}] $, are given in~\cite{sun2004multi}. It is known that~$ P_{k,ij} $ converges exponentially to steady-state values~$ P_{ij} $. Meanwhile, the weight matrix becomes~$W_{k,i}=W_i$. For simplicity, we denote~$ \hat{x}_{k,1:N} = [\hat{x}_{k,1}^\mathrm{T},\ldots,\hat{x}_{k,N}^\mathrm{T}]^\mathrm{T}$ and~{$W_{1:N} = [W_1, \ldots, W_N]$}.
We consider that the local estimators have entered steady state. The fusion estimation~\eqref{eqn:fusion x} becomes 
{
\begin{align} \label{eqn:fusion x 2}
	\bar{x}_k = W_{0} \hat{x}_{k,0} + W_{1:N} \hat{x}_{k,1:N}.
\end{align}}

\subsection{Data Packets with Disrupted Identities}

It is assumed that in each round of data fusion, the arrival of data packets in the network follows a predetermined distribution; that is, the arrangement of components~$ \hat{x}_{k,i} $ in~$ \hat{x}_{k,1:N} $ according to the order of arrival is governed by an \textit{a priori} probability distribution. Their identities are utilized by the fusion center to differentiate between various packets.
We now first introduce the concept of permutation and symmetry groups, characterizing all possible arrangements of data packets. Let~$G_N$ be an~$N$-order symmetric group with~$ N! $ ($ N $ factorial) different elements~$ g_i : \{1,\ldots,N\} \rightarrow \{1,\ldots,N\} $ which corresponds to a permutation matrix~$ \bar{\Gamma}_i \in \mathbb{R}^{N \times N}, 1 \leq i \leq N! $, and~$ \Gamma_{i} = \bar{\Gamma}_i \otimes I_n \in \mathbb{R}^{ nN \times nN} $. Moreover, we stipulate that~$ g_1 = \mathrm{id}$ is the identity mapping, i.e.,~$ g_1(i) = i $. Then, the received data packets on a specific case of arrangement, denoted as event~$ H_i $, can be represented as
{
\begin{align} \label{eqn: disorder i for x}
	\hat{x}^{[i]}_{k,1:N} = \Gamma_{i} \hat{x}_{k,1:N} = \left[ \hat{x}_{k,g_i(1)}^\mathrm{T}, \cdots , \hat{x}_{k,g_i(N)}^\mathrm{T} \right]^\mathrm{T},
\end{align}
where~$ \hat{x}^{[i]}_{k,1:N} $ is~$\hat{x}^*_{k,1:N}$ under the corresponding arrangement~$g_i$.}
{The probability of event~$ H_i $ is denoted as~$ \mathrm{Pr} ( H_i ) $, which is a constant.~$\{H_k\}$ is i.i.d. over time and independent of the physical process.}
Based on the properties of symmetric groups and permutation matrices, we can obtain that the permutation~$ g_i $ is invertible and closed. 
In other words, the correct local state estimations collected from the other sensors can be restored by~$ \hat{x}_{k,1:N} = \Gamma_{i}^\mathrm{T} \hat{x}^*_{k,1:N}|_{H_i} $ where~$ \exists \Gamma_j = \Gamma_i^\mathrm{T} $. 
\subsection{Problems of Interest}
Based on the above descriptions, the main problems we are interested in consist of the following:
\begin{enumerate}
	\item 
	How does the fusion center complete state estimation by utilizing the received data packets missing identity and the \textit{a priori} probability distribution of arrangements, i.e.,~$ \mathrm{Pr}(H_i) $?

	\item 
	What are the estimation performances of the developed algorithms for the previous problem?

	\item 
	If~$ \mathrm{Pr} (H_i) $ is unknown, how can the fusion center complete state estimation by utilizing only data packets missing identity? Can it converge correctly with the increase of received data?

\end{enumerate}

\section{Fusion Estimation with Prior Information} \label{section:The Bayesian Approach with P}

In this section, we first assume that the \textit{a priori} probability distribution of~$ H_i $ is known, proposing two data fusion mechanisms to handle disrupted identity. Then, we further analyze and compare the two approaches from the perspectives of fusion characteristics and estimation performance. The situation where the \textit{a priori} probability is unknown will be addressed in the next section.

\subsection{A Bayesian Approach} \label{section:A Bayesian Approach}
The received~$\hat{x}_{k,1:N}^*$ follows a Gaussian mixture model (GMM) because each arrangement is Gaussian distributed. Unlike a standard Bayesian estimator that directly evaluates the conditional expectation~\cite{flam2012mmse}, the proposed Bayesian approach only uses Bayesian posterior probabilities as fusion weights. However, the stability of the parameters on the distribution of~$\hat{x}_{k,1:N}^*$ cannot be ensured in unstable systems.
The computational complexity associated with determining the probability density of the received sequence~$ \hat{x}^*_{k,1:N} $ will increase over time, which is inconvenient for a module embedded in the fusion center to directly calculate the \textit{a posteriori} probability of~$ H_i $ by using~$ \hat{x}^*_{k,1:N} $. It encourages us to seek more convenient statistical characteristics as the condition for calculating~$\mathrm{Pr}(H_i)$. 

Since~$(A, C_i)$ is observable, we can conclude that the covariance of~$ x_k - \hat{x}_{k,i} $ converges. Therefore, the statistic~$ \theta_{k,i} = \hat{x}_{k,0} - \hat{x}_{k,i} $ has bounded and convergent covariance even for unstable system states. Let~$ \theta_{k} = [ \theta_{k,1}^\mathrm{T}, \cdots , \theta_{k,N}^\mathrm{T} ]^\mathrm{T} $. {The covariance between~$ \theta_{k,i} $ and~$ \theta_{k,j} $}, denoted as~$ \Psi_{ij} $, satisfies
\begin{align*} 
	\Psi_{ij} 
	= & \mathbb{E}[(\hat{x}_{k,i} - \hat{x}_{k,0}) (\hat{x}_{k,j} - \hat{x}_{k,0})^\mathrm{T}] \nonumber
	\\
	=& P_{00} + P_{ij} - P_{i0} - P_{0j}.
\end{align*} 
Obviously, it facilitates calculating the \textit{a posteriori} probability of~$ H_i $ due to the stability and boundedness of the covariance of~$ \theta_{k,i} $. Corresponding to~$ \hat{x}^*_{k,1:N} $, denote~$ \theta_{k} $ with unknown component arrangement as~$ \theta^*_{k} $. {The case corresponding to~$H_i$ is denoted as 
\begin{align} \label{eqn: def theta}
	\theta^{[i]}_{k} = \Gamma_i \theta_{k} = \left[ \theta_{k,g_i(1)}^\mathrm{T}, \cdots , \theta_{k,g_i(N)}^\mathrm{T} \right]^\mathrm{T}.
\end{align}
}
Since~$ \hat{x}_{k,0} $ is known for sensor~$ i = 0 $,~$ \theta_{k}^* $ can be calculated easily in the fusion estimation side.
The covariance matrices of~$ \theta_{k} $ and~$ \theta_{k}^*|_{H_i} $ can be given by~$ \Psi = \left[ \Psi_{ij} \right] \in \mathbb{R}^{nN \times nN} $ and~$ \Gamma_i \Psi \Gamma_i^\mathrm{T} $, respectively. We consider that~$ \Psi > 0 $. Since~$ \hat{x}_{k,i} $ is Gaussian and unbiased, the probability density function of~$ \theta^{[i]}_{k} $ is given by~$ \mathcal{N}_{\theta_{k}^{[i]}}(0, \Gamma_i \Psi \Gamma_i^\mathrm{T}) $.
From the law of total probability,~$ \theta_{k}^* $ follows a mixed Gaussian distribution
\begin{align} \label{eqn: distribution theta}
	p(\theta_{k}^*) = \sum_{i=1}^{N!}{\alpha_i \mathcal{N}_{\theta_{k}^*}(0, \Gamma_i \Psi \Gamma_i^\mathrm{T})},
\end{align}
where~$ \alpha_i = \mathrm{Pr}(H_i) $ is the \textit{a priori} probability of~$ H_i $. 
Let~$ \alpha^*_{k,i} = \mathrm{Pr}(H_i | \theta_{k}^*) $. By the Bayesian formula, we obtain that
\begin{align} \label{eqn: posterior probability}
	\alpha^*_{k,i} & =\frac{\mathrm{Pr}(H_i)}{p(\theta_k^*)} p(\theta_k^*|H_i) = \frac{\alpha_i e^{ -\frac{1}{2} (\theta_{k}^*)^\mathrm{T} \Gamma_i \Psi^{-1} \Gamma_i^\mathrm{T} \theta_{k}^* }}{\sum_{j=1}^{N!}{\alpha _j e^{ -\frac{1}{2} (\theta_{k}^*)^\mathrm{T} \Gamma_j \Psi^{-1} \Gamma_j^\mathrm{T} \theta_{k}^* } }}.
\end{align}
Given~$ H_i $, from~\eqref{eqn:fusion x 2} and~\eqref{eqn: disorder i for x}, the corresponding estimation for fusion process is
\begin{align} \label{eqn: recover with H_i}
	\bar{x}_k = W_{0} \hat{x}_{k,0} + W_{1:N} \Gamma_{i}^\mathrm{T} \hat{x}^{[i]}_{k,1:N}.
\end{align}
Therefore, using the Bayesian posterior probabilities as fusion weights, the estimated value of~$ x_k $ can be obtained as follows:
\begin{align*}
	 \bar{x}_k^B =& \sum_{i=1}^{N!}{\mathrm{Pr}(H_i|\theta_k^*) \cdot \left( W_0 \hat{x}_{k,0} + W_{1:N} \Gamma_i^\mathrm{T} \hat{x}^*_{k,1:N} \right) } \nonumber
	\\
	=& W_0 \hat{x}_{k,0} + W_{1:N} \left( \sum_{i=1}^{N!}{\alpha^*_{k,i} \Gamma_{i}^\mathrm{T}}\right) \cdot \hat{x}^*_{k,1:N}.
\end{align*}
For simplicity, denote the mixed permutation matrix as~$ \Gamma_{\alpha^*_k} = \sum_{i=1}^{N!}{\alpha^*_{k,i} \Gamma_{i}} $. The Bayesian fusion estimation becomes
\begin{align} \label{eqn: Bayesian estimation}
	\bar{x}_k^{B} = W_0 \hat{x}_{k,0} + W_{1:N} \Gamma_{\alpha^*_k}^\mathrm{T} \hat{x}^*_{k,1:N}.
\end{align}
For a concise representation of the quadric error matrix for the Bayesian approach, we define~$ P^B_k = \mathbb{E}[(x_k - \bar{x}_k^B)(x_k - \bar{x}_k^B)^\mathrm{T}]$. Let the estimation error of sensor~$ i $ be~$ e_{k,i} = x_k - \hat{x}_{k,i} $. We denote~$ \tilde{\Psi}_i $ as
\begin{align} \label{eqn: tilde Psi}
	\tilde{\Psi}_i = \mathbb{E} \left[ \theta _{k}^{*}e_{k,0}^{\mathrm{T}}|H_i \right] = \mathbb{E} [ \theta _{k}^{[i]}e_{k,0}^{\mathrm{T}} ] = \Gamma_i \tilde{\Psi}_1.
\end{align}
The following theorem provides the expression of~$P^B_k$.
\begin{theorem} 
	For the system \eqref{eqn:process1}--\eqref{eqn:standard P} where the identity of data packets is disrupted, the quadric error matrix of the fusion center with the Bayesian approach is given by 
	\begin{align} \label{eqn: P_B}
		P^B_k =&  \sum_{i=1}^{N!} { \alpha_i W_{1:N} \tilde{\Theta}_i^B W_{1:N}^\mathrm{T} } + \sum_{i=1}^{N!} { \alpha_i W_{1:N} \Theta_i^B \Gamma _i\Psi^{-1} \tilde{\Psi}_{1} } \nonumber
		\\
		& + \sum_{i=1}^{N!} { \alpha_i \tilde{\Psi}_{1}^\mathrm{T} \Psi^{-1} \Gamma _{i}^{\mathrm{T}} (\Theta_i^B)^\mathrm{T} W_{1:N}^\mathrm{T} } + P_{00},
	\end{align}
	where the constant matrices~$ \tilde{\Theta}_i^B $ and~$ \Theta_i^B $ are given by  
	\begin{align} 
		\Theta_i^B =& \int_{\mathbb{R}^{nN}} \Gamma_{\alpha^*_k}^\mathrm{T} \theta_k^* (\theta_k^*)^\mathrm{T}  \mathcal{N}_{\theta_{k}^*}(0, \Gamma_i \Psi \Gamma_i^\mathrm{T}) \mathrm{d} \theta_{k}^*, \label{eqn: Theta_i B}
		\\
		\tilde{\Theta}_i^B =& \int_{\mathbb{R}^{nN}} \Gamma_{\alpha^*_k}^\mathrm{T} \theta_k^* (\theta_k^*)^\mathrm{T} \Gamma_{\alpha^*_k} \mathcal{N}_{\theta_{k}^*}(0, \Gamma_i \Psi \Gamma_i^\mathrm{T}) \mathrm{d} \theta_{k}^*. \label{eqn: tilde Theta_i B}
	\end{align}
	Furthermore, an upper bound of~$ \mathrm{Tr}\{P^B_k\} $ is given by
	\begin{align} \label{eqn: upper bound P_B}
		\mathrm{Tr} \{ P^B_k \} \leq & \mathrm{Tr} \{ P_{00} + \tilde{\Psi}_1^{\mathrm{T}} \Psi^{-1} \tilde{\Psi}_1 \} \nonumber
		\\
		& + 2 \left\| W_{1:N}^\mathrm{T} W_{1:N} \right\|_{\mathrm{F}} \mathrm{Tr}\left\{ \Psi \right\}.
	\end{align}
	\begin{pf*}{\bf Proof.}
		Note that~$ \sum_{i=1}^{N!}{\alpha^*_{k,i}} = 1 $ and~$ \sum_{i=0}^{N}{W_i} = I $. It can be obtained that~$ I_n = W_0 + W_{1:N} \Gamma_{\alpha^*_k}^\mathrm{T} I_{n,N} $. Given~$H_i$, the received data packets are~$ \hat{x}^{*}_{k,1:N} = \hat{x}^{[i]}_{k,1:N} $. From~\eqref{eqn: Bayesian estimation}, the estimation error becomes 
		\begin{align} \label{eqn: estimation error Bayes H 1}
			x_k - \bar{x}_k^B = e_{k,0} + W_{1:N} \Gamma_{\alpha^*_k}^\mathrm{T} \theta_{k}^{[i]}.
		\end{align}
		Note that the covariance between~$ e_{k,0} $ and~$ \theta_{k}^{[i]} $ is provided by~$ \tilde{\Psi}_i = \Gamma_i \tilde{\Psi}_1 $. Let~$ \tilde{e}_{k,0}=e_{k,0}-\tilde{\Psi}_{i}^{\mathrm{T}}\Gamma _i\Psi^{-1} \Gamma _{i}^{\mathrm{T}} \theta_{k}^{[i]} $. The joint distribution of~$ \tilde{e}_{k,0} $ and~$ \theta_{k}^{[i]} $ given~$ H_i $ becomes
		\begin{align} \label{eqn: distribution tilde e theta}
			\left[ \begin{array}{c}
				\tilde{e}_{k,0}\\
				\theta_{k}^{[i]}\\
			\end{array} \right] \sim \mathcal{N} \left(0, \left[ \begin{matrix}
				P_{00}-\tilde{\Psi}_{1}^{\mathrm{T}} \Psi^{-1} \tilde{\Psi}_1&		\\
				&		\Gamma _i\Psi \Gamma _{i}^{\mathrm{T}}\\
			\end{matrix} \right] \right).
		\end{align}
		In this view,~$ x_k - \bar{x}_k^B $ becomes
		\begin{align}
			x_k - \bar{x}_k^B = \tilde{e}_{k,0} + \left( W_{1:N} \Gamma_{\alpha^*_k}^\mathrm{T} + \tilde{\Psi}_{1}^{\mathrm{T}} \Psi^{-1} \Gamma _{i}^{\mathrm{T}} \right) \theta_{k}^{[i]}.
		\end{align}
		From \eqref{eqn: tilde Psi} and \eqref{eqn: distribution tilde e theta}, the quadric error matrix given~$ H_i $ is 
		\begin{align} \label{eqn: PB_k H_i}
			& \mathbb{E}[(x_k - \bar{x}_k^B)(x_k - \bar{x}_k^B)^\mathrm{T} | H_i] \nonumber
			\\
			=& P_{00} + W_{1:N} \mathbb{E}[\Gamma_{\alpha^*_k}^\mathrm{T} \theta_k^* (\theta_k^*)^\mathrm{T} \Gamma_{\alpha^*_k} | H_i ] W_{1:N}^\mathrm{T} \nonumber
			\\
			& + W_{1:N} \mathbb{E}[\Gamma_{\alpha^*_k}^\mathrm{T} \theta_k^* (\theta_k^*)^\mathrm{T} | H_i ] \Gamma _i\Psi^{-1} \tilde{\Psi}_{1} \nonumber
			\\
			& + \tilde{\Psi}_{1}^\mathrm{T} \Psi^{-1} \Gamma _{i}^{\mathrm{T}} \mathbb{E}[ \theta_k^* (\theta_k^*)^\mathrm{T} \Gamma_{\alpha^*_k} | H_i ]  W_{1:N}^\mathrm{T},
		\end{align}
		where the matrices~$ \Theta_i^B = \mathbb{E}[\Gamma_{\alpha^*_k}^\mathrm{T} \theta_k^* (\theta_k^*)^\mathrm{T} | H_i ]  $ and~$ \tilde{\Theta}_i^B = \mathbb{E}[\Gamma_{\alpha^*_k}^\mathrm{T} \theta_k^* (\theta_k^*)^\mathrm{T} \Gamma_{\alpha^*_k} | H_i ] $ are obtained by \eqref{eqn: Theta_i B} and \eqref{eqn: tilde Theta_i B}, respectively. Then, according to the total probability criterion, the quadric error matrix satisfies the following identity:
		\begin{align} \label{eqn: total probability criterion}
			P^B_k = \sum_{i=1}^{N!}{\alpha_i \mathbb{E}[(x_k - \bar{x}_k^B)(x_k - \bar{x}_k^B)^\mathrm{T} | H_i] }.
		\end{align}
		From \eqref{eqn: PB_k H_i} and \eqref{eqn: total probability criterion}, \eqref{eqn: P_B} holds. 
		
		Next, we derive the upper bound of~$ P^B_k $. Note that
		\begin{align}
			\mathbb{E} \left[ \left( W_{1:N} \Gamma_{\alpha^*_k}^\mathrm{T} - \tilde{\Psi}_{i}^{\mathrm{T}}\Gamma _i\Psi^{-1} \Gamma _{i}^{\mathrm{T}} \right) \theta_{k}^* \left( \cdot \right)^\mathrm{T} | {H_i} \right] \geq 0.
		\end{align}
		The second and third terms in \eqref{eqn: PB_k H_i} satisfy the following inequality:
		\begin{align} \label{eqn: W Psi matrix upper bound}
			& W_{1:N} \mathbb{E}[\Gamma_{\alpha^*_k}^\mathrm{T} \theta_k^* (\theta_k^*)^\mathrm{T} | H_i ] \Gamma _i\Psi^{-1} \tilde{\Psi}_{1} + \left( \cdot \right)^\mathrm{T} \nonumber
			\\
			\leq & W_{1:N} \tilde{\Theta}_i^B W_{1:N}^\mathrm{T} + \tilde{\Psi}_1^{\mathrm{T}} \Psi^{-1} \tilde{\Psi}_1.
		\end{align}
		We now consider the terms of~$ W_{1:N} \tilde{\Theta}_i^B W_{1:N}^\mathrm{T} $. Define~$J(\theta_k^*, \Gamma_{\alpha^*_k}) = (\theta_k^*)^\mathrm{T} \Gamma_{\alpha^*_k} W_{1:N}^\mathrm{T} W_{1:N} \Gamma_{\alpha^*_k}^\mathrm{T} \theta_k^*$
		and let~$ \mathrm{d} \mu_{\theta_{k}^{[i]}} = \mathcal{N}_{\theta_{k}^{[i]}}(0, \Gamma_i \Psi \Gamma_i^\mathrm{T}) \mathrm{d} \theta_{k}^{[i]} $ be the probability measure. Then, we obtain 
		\begin{align} \label{eqn: trace W theta W}
			\mathrm{Tr} \left\{ W_{1:N} \tilde{\Theta}_i^B W_{1:N}^\mathrm{T} \right\} = \int J(\theta_{k}^{[i]}, \Gamma_{\alpha^*_k}) \mathrm{d} \mu_{\theta_{k}^{[i]}}.
		\end{align}
		Since~$ \Gamma_{\alpha^*_k} $ is a doubly stochastic matrix, according to the convexity of positive definite quadratic forms, we have~$\| \Gamma_{\alpha^*_k} W_{1:N}^\mathrm{T} W_{1:N} \Gamma_{\alpha^*_k}^\mathrm{T} \|_{\mathrm{F}} \leq \| W_{1:N}^\mathrm{T} W_{1:N} \|_{\mathrm{F}}$. Therefore, an upper bound of~$ J(\theta_k^*, \Gamma_{\alpha^*_k}) $ can be represented as~$J(\theta_k^*, \Gamma_{\alpha^*_k}) \leq \left\| W_{1:N}^\mathrm{T} W_{1:N} \right\|_{\mathrm{F}} \cdot \left\| \theta_k^* \right\|^2$. Combining this bound with~\eqref{eqn: trace W theta W}, we obtain  
		\begin{align} \label{eqn: upper bound trace W theta W}
			\mathrm{Tr} \left\{ W_{1:N} \tilde{\Theta}_i^B W_{1:N}^\mathrm{T} \right\} \leq \left\| W_{1:N}^\mathrm{T} W_{1:N} \right\|_{\mathrm{F}} \mathrm{Tr} \left\{  \Psi  \right\}.
		\end{align}
		Combining \eqref{eqn: PB_k H_i}, \eqref{eqn: total probability criterion}, \eqref{eqn: W Psi matrix upper bound}, and \eqref{eqn: upper bound trace W theta W}, we obtain the upper bound of~$ \mathrm{Tr} \{ P^B_k \} $ in \eqref{eqn: upper bound P_B}.
	\end{pf*}
\end{theorem} 

\subsection{A Greedy Approach} \label{section:A Greedy Approach}
Another approach is to guess the arrangement of the data packets first, i.e., event~$ H_i $, and then fuse the received estimation results in this arrangement by~\eqref{eqn: recover with H_i}, which we call the greedy approach.
As is well known,~$ \theta_{k}^\mathrm{T} \Psi^{-1} \theta_{k} \sim \chi^2(nN) $. If the arrangement of the components~$ \theta_{k,i} $ in~$ \theta_{k} $ is disrupted, this proposition no longer holds. More importantly, in statistical terms,~$ (\theta_{k}^*)^\mathrm{T} \Psi^{-1} \theta_{k}^* $ tends towards a larger mean~\cite{zhao2025state}; 
that is,~$ \mathbb{E}[(\theta_{k}^*)^\mathrm{T} \Psi^{-1} \theta_{k}^*|H_i] \geq \mathbb{E}[\theta_{k}^\mathrm{T} \Psi^{-1} \theta_{k}]$.
Therefore, we can obtain that
\begin{align*} 
	i = \underset{ 1 \leq \hat{i} \leq N! } {arg \min}  ~\mathbb{E}\left[(\Gamma_{\hat{i}}^\mathrm{T}\theta_{k}^*)^\mathrm{T} \Psi^{-1} \Gamma_{\hat{i}}^\mathrm{T} \theta_{k}^*| H_i \right].
\end{align*}
In other words, when the guessed event~$ H_{\hat{i}} $ is exactly~$ H_i $, the term~$ (\Gamma_{\hat{i}}^\mathrm{T}\theta_{k}^*)^\mathrm{T} \Psi^{-1} \Gamma_{\hat{i}}^\mathrm{T} \theta_{k}^* $ can be minimized in the expectation sense.
In this view, we construct the greedy approach to guess an event~$ H_{\hat{i}} $ by the likelihood decision
\begin{align} \label{eqn: guess H}
	\hat{i} = \underset{ 1 \leq i \leq N! } {arg \max}~p(\theta_k^*|H_i)
\end{align}
to recover the arrangement of~$ \theta_{k,i} $ in~$ \theta_{k}^* $. Then, the algorithm in the fusion center with the greedy approach is
\begin{align} \label{eqn: estimation guess}
	\bar{x}_k^G = W_{0} \hat{x}_{k,0} + W_{1:N} \Gamma_{\hat{i}}^\mathrm{T} \hat{x}^*_{k,1:N}.
\end{align}
The estimation-error second-moment matrix of~$\bar{x}_k^G$ is denoted as~$P^G_k$, and the following theorem provides its expression.
\begin{theorem}
	For the system \eqref{eqn:process1}--\eqref{eqn:standard P} where the identity of data packets is disrupted, the quadric error matrix of the fusion center with the greedy approach is given by 
	\begin{align} \label{eqn: P_G}
		P^G_k =&  \sum_{i=1}^{N!} { \alpha_i W_{1:N} \tilde{\Theta}_i^G W_{1:N}^\mathrm{T} } + \sum_{i=1}^{N!} { \alpha_i W_{1:N} \Theta_i^G \Gamma _i\Psi^{-1} \tilde{\Psi}_{1}  } \nonumber
		\\
		& + \sum_{i=1}^{N!} { \alpha_i \tilde{\Psi}_{1}^\mathrm{T} \Psi^{-1} \Gamma _{i}^{\mathrm{T}} (\Theta_i^G)^\mathrm{T} W_{1:N}^\mathrm{T} } + P_{00},
	\end{align}
	where the constant matrices~$ \tilde{\Theta}_i^G $ and~$ \Theta_i^G $ are given by 
	\begin{align}
		\Theta_i^G =& \sum_{j=1}^{N!}{\int_{\Omega_j} \Gamma_j^\mathrm{T} \theta_k^* (\theta_k^*)^\mathrm{T}  \mathcal{N}_{\theta_{k}^*}(0, \Gamma_i \Psi \Gamma_i^\mathrm{T}) \mathrm{d} \theta_{k}^*}, \label{eqn: Theta_i G}
		\\
		\tilde{\Theta}_i^G =& \sum_{j=1}^{N!}{\int_{\Omega_j} \Gamma_j^\mathrm{T} \theta_k^* (\theta_k^*)^\mathrm{T} \Gamma_j \mathcal{N}_{\theta_{k}^*}(0, \Gamma_i \Psi \Gamma_i^\mathrm{T}) \mathrm{d} \theta_{k}^*}. \label{eqn: tilde Theta_i G}
	\end{align}
	The set~$ \Omega_i $ is defined as 
	\begin{align*}
		\Omega_i = \bigcap_{j=1}^{N!}\bigg\{ \theta _{k}^{*}\in \mathbb{R} ^{nN}| & \left( \theta _{k}^{*} \right) ^{\mathrm{T}}\Gamma _i\Psi^{-1} \Gamma _{i}^{\mathrm{T}}\theta _{k}^{*} 
		\\
		& \le \left( \theta _{k}^{*} \right) ^{\mathrm{T}}\Gamma _j\Psi^{-1} \Gamma _{j}^{\mathrm{T}}\theta _{k}^{*} \bigg\}.
	\end{align*}
	Furthermore, an upper bound of~$ \mathrm{Tr}\{P^G_k\} $ is given by
	\begin{align*}
	\mathrm{Tr} \{ P^G_k \} \leq & \mathrm{Tr} \{ P_{00} + \tilde{\Psi}_1^{\mathrm{T}} \Psi^{-1} \tilde{\Psi}_1 \} \nonumber
		\\
		& + 2 \left\| W_{1:N}^\mathrm{T} W_{1:N} \right\|_{\mathrm{F}} \mathrm{Tr}\left\{ \Psi \right\}.
		\end{align*}
	\begin{pf*}{\bf Proof.}
		Similar to~\eqref{eqn: distribution tilde e theta}--\eqref{eqn: PB_k H_i}, the quadric error matrix can be written as 
		\begin{align}
			&  \mathbb{E}[(x_k - \bar{x}_k^{G})(x_k - \bar{x}_k^{G})^\mathrm{T} | H_i] \nonumber
			\\
			=& P_{00} + W_{1:N} \mathbb{E}[\Gamma_{\hat{i}}^\mathrm{T} \theta_k^* (\theta_k^*)^\mathrm{T} \Gamma_{\hat{i}} | H_i ] W_{1:N}^\mathrm{T} \nonumber
			\\
			& + W_{1:N} \mathbb{E}[\Gamma_{\hat{i}}^\mathrm{T} \theta_k^* (\theta_k^*)^\mathrm{T} | H_i ] \Gamma_i \Psi^{-1} \tilde{\Psi}_{1} \nonumber
			\\
			& + \tilde{\Psi}_{1}^\mathrm{T} \Psi^{-1} \Gamma_{i}^{\mathrm{T}} \mathbb{E}[ \theta_k^* (\theta_k^*)^\mathrm{T} \Gamma_{\hat{i}} | H_i ]  W_{1:N}^\mathrm{T},
		\end{align}
		where~$ \hat{i} $ is given by \eqref{eqn: guess H}. Therefore, we can infer that 
		\begin{align*}
			\Theta_i^G = \mathbb{E}[\Gamma_{\hat{i}}^\mathrm{T} \theta_k^* (\theta_k^*)^\mathrm{T} | H_i ]  \text{ and } \tilde{\Theta}_i^G = \mathbb{E}[\Gamma_{\hat{i}}^\mathrm{T} \theta_k^* (\theta_k^*)^\mathrm{T} \Gamma_{\hat{i}} | H_i ].
		\end{align*}
		The rest of the proofs are similar to \eqref{eqn: total probability criterion}--\eqref{eqn: upper bound trace W theta W}.		
	\end{pf*}
\end{theorem}
\begin{remark}
	In fact, the proposed data fusion method is still applicable to scenarios with data packet loss. Under conditions of packet loss where the number of received packets is less than the number of data sources,~$\theta_k^*$ still follows a Gaussian mixture distribution. The number of Gaussian mixture components corresponds to the number of possible permutations of the received packets.  In such cases, both the Bayesian and greedy approaches can be directly adapted from the full-data scenario. The only difference is that the weight matrix~$W$ needs to be recalculated based on the data source matched to the data packet.
\end{remark}

\subsection{Analysis and Comparison}
In fact, the Bayesian based fusion estimation can be regarded as an approximately optimal
posterior weighted estimation. Specifically, for the optimal perspective, we have
\begin{align} \label{eqn:optimal Bayesian}
    \bar{x}_k^{*} 
    =& \sum_{i=1}^{N!} \alpha^*_{k,i}
    (W_0 \hat{x}_{k,0}
    + W_{1:N} \Gamma_{i}^\mathrm{T} \hat{x}^*_{k,1:N})
    \nonumber\\
    =& \sum_{i=1}^{N!} \alpha^*_{k,i} \cdot 
    \underset{\bar{x} \in \mathcal{F}_i}{\arg\min}
    ~\mathbb{E}
    \left[
    \|x_k-\bar{x}\|^2 \mid H_i
    \right],
\end{align}
where~$\mathcal{F}_i$ denotes the class of linear unbiased fusion
estimators under arrangement~$H_i$.
Since~$\alpha^*_{k,i}$ is expected to be close to~$\mathrm{Pr}(H_i| \hat{x}_{k,0}, \theta_k^*)$, the resulting~\eqref{eqn: Bayesian estimation} can be regarded as an approximation to the optimal weighted fusion.
However, it is not a satisfactory strategy in all scenarios. Generally speaking, the fusion center cannot restore the received data to its original state in any system time-step. More specifically, if~$ \exists i_1, i_2 \in \{ 1, \ldots, N! \} $ such that~$ \alpha^*_{k,i_1} \cdot \alpha^*_{k,i_2} > 0 $ and~$ \hat{x}_{k,i_1} \ne \hat{x}_{k,i_2} $, then~$ \hat{x}_{k,1:N} \ne \Gamma_{\alpha^*_k}^\mathrm{T} \hat{x}^*_{k,1:N} $. 
Compared to the Bayesian approach, the greedy method may recover data without loss, i.e.,~$ \hat{x}_{k,1:N} = \Gamma_{\hat{i}}^\mathrm{T} \hat{x}^*_{k,1:N} $ if~$ H_{\hat{i}} $ guesses right.
In fact, from \eqref{eqn: posterior probability}, we notice that the mechanism \eqref{eqn: guess H} is a likelihood decision such that
\begin{align} \label{eqn:optimal greedy}
	\hat{i} = \underset{1 \le i \le N!}{arg \max}~p(\theta_k^*|H_i).
\end{align}
In other words, it guesses the event~$ H_i $ with the largest likelihood under the received data packets, recovering data as much as possible.
Both developed methods correspond to different fusion mechanisms. 

\section{Learning Prior Information with the EM Algorithm} \label{section:Learning Prior Information with the EM Algorithm}

This section considers the case in which the \textit{a priori} probabilities of events~$ H_i $ are unknown constants. Our goal is to address situations lacking prior information by embedding the EM algorithm into the given fusion framework.
\subsection{Identity Estimation Based on EM Algorithm}

We initially consider datasets generated by probability density functions in a linear combination form. Subsequently, we apply this result to the dataset consisting of Gaussian mixture distribution data, which can prove the correct convergence of the identification of~$\mathrm{Pr}(H_i)$. Let~$ ( \mathcal{X}, \mathscr{B}, \mu) $ be a probability space where~$\mathscr{B}$ is the Borel~$\sigma$-algebra of~$\mathcal{X}$, and~$ \mu $ is the corresponding probability measure determined by~$ \mathrm{d} \mu = \bar{f}_\omega d \xi_i $ with~$\xi_i \in \mathcal{X}$. We investigate the probability density function such that~$ \bar{f}_\omega = \sum_{i=1}^{N_f} \omega_i f_i,$
where~$ f_1, \ldots, f_{N_f} $ are linearly independent probability density functions, and~$ \omega = [\omega_1, \ldots, \omega_{N_f}] $ satisfies~$ \sum_{i=1}^{N_f} { \omega_i } = 1 $ and~$ \omega_i > 0 $.
Denote~$ \Xi_L = \{ \xi_i \}_{i=1}^{L}$ as a set of random vectors produced by~$ \xi_i \sim \bar{f}_\omega(\xi_i) $. Obviously, the considered probability density~$\bar{f}_\omega$ contains the Gaussian mixture distribution.

For the EM algorithm~\cite{watanabe2010interval}, the dataset~$ \Xi_L $ from a certain sampling is considered as input to the algorithm.~$ \xi_i $ can be viewed as generated by the i.i.d. hidden variable~$ \beta_i $ with the conditional distribution~$	\xi_i |_{\beta_i = j} \sim f_j(\xi_i)$ and~$\mathrm{Pr}(\beta_i = j) = \omega_j$.
In this situation,~$ \xi_i $ still satisfies
\begin{align*}
	\xi_i \sim \sum_{j=1}^{N_f} \mathrm{Pr}(\beta_i = j) f_j(\xi_i) = \bar{f}_\omega(\xi_i).
\end{align*}
We assume that the parameter~$ \omega $ is unknown, and the functions~$ f_i $ are known.
Denote~$ \beta_{i,j}^t $ as the estimate at the~$t$th iteration of the probability for the~$i$th sample generated by the~$j$th component~$ f_j $. When~$ \xi_i $ are i.i.d., the E-step is given by 
\begin{align} \label{eqn: E step}
	\beta_{i,j}^t =& Pr( \beta_i = j | \xi_i, \hat{\omega}^t) =  Pr( \beta_i = j | \hat{\omega}^t) \frac{p( \xi_i | \beta_i = j, \hat{\omega}^t) }{p( \xi_i | \hat{\omega}^t) }  \nonumber
	\\
	=& \frac{\hat{\omega}_j^t f_j(\xi_i)}{\sum_{n=1}^{N_f} \hat{\omega}_n^t f_n(\xi_i)} = \frac{\hat{\omega}_j^t f_j(\xi_i)}{\bar{f}_{\hat{\omega}^t}(\xi_i)}.
\end{align}
For the M-step, let~$Q( \omega | \hat{\omega}^t ) =  \sum_{i=1}^L { \sum_{j=1}^{N_f} { \beta_{i,j}^t \log [\omega_j f_j(\xi_i)] }}$. We have
\begin{align} 
	\hat{\omega}_j^{t+1} &= \underset{\omega_j} {arg \max}{ ~Q( \omega | \hat{\omega}^t ) } = \frac{1}{L} \sum_{i=1}^L \beta_{i,j}^t, \label{eqn: M step}
\end{align}
where~$ \hat{\omega}_i^{t} = \hat{\omega}_i^{t} (\Xi_L) $ can be regarded as a single-valued function of the data set~$\Xi_L$, while~$ \hat{\omega}^{t} = [\hat{\omega}_1^{t},\ldots, \hat{\omega}_{N_f}^{t}] $ is the estimate at the~$t$th iteration of~$ \omega $. 

In the fusion center, we can embed a module to estimate~$\mathrm{Pr}(H_i)$ using data~$ \theta_k^* $ collected in real time. The algorithm~\eqref{eqn: E step}--\eqref{eqn: M step} applied to the fusion estimation is
\begin{align}
	&\beta_{i,j}^t = \frac{\hat{\alpha}_j^t e^{ -\frac{1}{2} (\theta_{i}^*)^\mathrm{T} \Gamma_j \Psi^{-1} \Gamma_j^\mathrm{T} \theta_{i}^* }}{\sum_{n=1}^{N!}{\hat{\alpha}_n^t e^{ -\frac{1}{2} (\theta_{i}^*)^\mathrm{T} \Gamma_n \Psi^{-1} \Gamma_n^\mathrm{T} \theta_{i}^* } }}, \label{eqn:fusion EM 1}
	\\
	&\hat{\alpha}_j^{t+1} = \frac{1}{L} \sum_{i=1}^{ L } \beta_{i,j}^t. \label{eqn:fusion EM 2}
\end{align}

\subsection{Correct Convergence Analysis}
Common available literature about the EM algorithm generally requires data to be i.i.d. and may not guarantee correct convergence for the identification of parameters. We next use ergodic theory to establish convergence to the true parameters under the assumptions stated below.
We first summarize the following statements.
\begin{lemma} \label{lemma: ergodic process}
	For a measure-preserving transformation~$ T : \mathcal{X} \rightarrow \mathcal{X} $ on a measurable space~$ ( \mathcal{X}, \mathscr{B}, \mu) $ and for~$ \xi \in  \mathcal{X} $, the following statements hold true:
	\begin{enumerate}
		\item 
		{\rm(Ergodicity)}~
		$ T $ is ergodic if~$ \forall B \in \mathscr{B} $,~$T^{-1} B = B \Rightarrow \mu(B) \in \{0,1\}$, while the process~$ \{ T^i \xi \}_{i=0}^\infty $ is ergodic.
		
		\item
		{\rm(Birkhoff Ergodic Theorem~\cite{klenke2020probability})}
		For any~$ \varphi \in L^1$-space with respect to the measure~$ \mu $, and given that~$ T $ is ergodic, we have~$ \underset{n \rightarrow \infty} \lim \frac{1}{n} \sum_{i=0}^{n-1} { \varphi \left[ T^i(\xi) \right] } = \int \varphi \mathrm{d} \mu$.
	\end{enumerate}
\end{lemma}
We now provide the following proposition to establish the correct convergence in our situation.
\begin{proposition} \label{proposition: correct converge of algorithm}
	For the algorithm \eqref{eqn: E step}--\eqref{eqn: M step} based on the data set~$ \Xi_L = \{ \xi_i \}_{i=1}^{L} $, if~$ \{\xi_i\}_{i=0}^\infty $ is ergodic,~$ \hat{\omega}^{t} $ converges in probability to
	\begin{align} \label{eqn: limit t L}
		\underset{t \to \infty} \lim \underset{L \to \infty} \lim \hat{\omega}^{t+1}_j (\Xi_L) =  \underset{L \to \infty} \lim \underset{t \to \infty} \lim \hat{\omega}^{t+1}_j  (\Xi_L) = \omega,
	\end{align}
	where the initial value of~$ \hat{\omega}^{t} $ satisfies~$ \sum_{i=1}^{N_f} { \hat{\omega}_i^0 } = 1 $ and~$\hat{\omega}_i^0 > 0$.
	\begin{pf*}{\bf Proof.}
		{
		From \eqref{eqn: E step} and \eqref{eqn: M step}, since~$ \{\xi_i\}_{i=0}^\infty $ is ergodic, by the Birkhoff ergodic theorem,~$ \hat{\omega}^{t+1}_j $ converges in probability to 
		\begin{align} \label{eqn: L infinity}
			\underset{t \to \infty} \lim \underset{L \to \infty} \lim \hat{\omega}^{t+1}_j (\Xi_L) = \underset{t \to \infty} \lim \int \frac{\hat{\omega}_j^t f_j}{\bar{f}_{\hat{\omega}^t}} \mathrm{d} \mu.
		\end{align}
		As is well known, due to the monotonic optimization of the EM algorithm, it will definitely converge for given~$L$, that is,
		\begin{align} \label{eqn: limit L t infinity}
			\underset{L \to \infty} \lim \underset{t \to \infty} \lim \hat{\omega}^{t+1}_j (\Xi_L) = \int \underset{t \to \infty} \lim \frac{\hat{\omega}_j^t f_j}{\bar{f}_{\hat{\omega}^t}} \mathrm{d} \mu.
		\end{align}
		Since~$ \hat{\omega}_j^t f_j / \bar{f}_{\hat{\omega}^t} \leq 1 $, the interchange of the limit and the integral
		\begin{align}
			\underset{t \to \infty} \lim \int \frac{\hat{\omega}_j^t f_j}{\bar{f}_{\hat{\omega}^t}} \mathrm{d} \mu = \int \underset{t \to \infty} \lim \frac{\hat{\omega}_j^t f_j}{\bar{f}_{\hat{\omega}^t}} \mathrm{d} \mu
		\end{align}
		is given by Lebesgue's dominated convergence theorem~\cite{rudin1964principles}. We now prove that}
		\begin{align} \label{eqn: omega is limit of t L}
			\underset{ t \to \infty } \lim \underset{ L \to \infty } \lim \hat{\omega}^{t+1}_j (\Xi_L) = \omega.
		\end{align}
		Let~$ \underset{ t \to \infty } \lim \hat{\omega}^{t+1}_j = \omega^* $. From \eqref{eqn: L infinity}, we have~$\int f_i/ \bar{f}_{\omega^*} \mathrm{d} \mu = 1, ~\forall i \in \{1,\ldots,N_f \}$.
		Therefore, we can obtain that
		\begin{align} 
			&\omega_1 + \ldots + \omega_{N_f} = \int \frac{\bar{f}_\omega}{\bar{f}_{\omega^*}} \mathrm{d} \mu = 1, \label{eqn: f/f 1}
			\\
			&\omega_1^* + \ldots + \omega_{N_f}^* = \int \frac{\bar{f}_{\omega^*}}{\bar{f}_\omega} \mathrm{d} \mu = 1. \label{eqn: f/f 2}
		\end{align}
		From \eqref{eqn: f/f 1} and \eqref{eqn: f/f 2}, we derive that 
		\begin{align}
			\int \left( \frac{\bar{f}_\omega}{\bar{f}_{\omega^*}} + \frac{\bar{f}_{\omega^*}}{\bar{f}_\omega} - 2 \right) \mathrm{d} \mu = 0.
		\end{align}
		Therefore, we have~$ \bar{f}_{\omega^*} = \bar{f}_\omega $ almost everywhere with respect to the measure~$ \mu $ and~$ \omega^* = \omega $. 
		So far, we have proved~\eqref{eqn: omega is limit of t L}.
		In addition, we see that the functions~$ \hat{\omega}^{1} (\Xi_L), \ldots, \hat{\omega}^{\infty} (\Xi_L) $ are well-defined, bounded, and continuous for given~$ \Xi_1, \ldots, \Xi_\infty $. The proof is completed.
	\end{pf*}
\end{proposition}

We now demonstrate the correct convergence in dealing with the identification process at the fusion center. 

\begin{theorem}
	Let the dataset~$ \Xi_L = \{ \theta_i^* \}_{i=1}^{L} $ be generated by the system \eqref{eqn:process1}--\eqref{eqn:fusion x}. Assume that the Gaussian densities~$\{\mathcal{N}_{\theta}(0,\Gamma_i\Psi\Gamma_i^{T})\}_{i=1}^{N!}$ are linearly independent, ensuring identifiability of \(\alpha\). The algorithm \eqref{eqn:fusion EM 1}--\eqref{eqn:fusion EM 2} for identifying parameters~$ \hat{\alpha}^t $ in \eqref{eqn: distribution theta} converges in probability to
	\begin{align}
		\underset{t \to \infty} \lim \underset{L \to \infty} \lim \hat{\alpha}^{t} (\Xi_L)  = \underset{L \to \infty} \lim \underset{t \to \infty} \lim \hat{\alpha}^{t} (\Xi_L) = \alpha,
	\end{align} 
	where~$ \alpha = [\alpha_1, \ldots, \alpha_{N!} ] $ and~$ \hat{\alpha}^{t} = [\hat{\alpha}_1^{t},\ldots, \hat{\alpha}_{N!}^{t}] $,~$ \sum_{i=1}^{N!} { \hat{\alpha}_i^0 } = 1, \hat{\alpha}_i^0 > 0 $. 
	\begin{pf*}{\bf Proof.}
		Let~$ \xi_k = \theta_{k}^* $ and~$ \omega_i = \alpha_i $. The probability density functions~$ f_i $ and~$ \bar{f}_\omega $ are given by~$f_i ( \theta_{k}^* ) = \mathcal{N}_{\theta_{k}^*}(0, \Gamma_i \Psi \Gamma_i^\mathrm{T})$ and~$ \bar{f}_\omega (\theta_{k}^*) = \sum_{i=1}^{N!}{\alpha_i \mathcal{N}_{\theta_{k}^*}(0, \Gamma_i \Psi \Gamma_i^\mathrm{T})}$, respectively. 
		We now verify that~$ \theta_{k}^* $ is ergodic. Let~$ T $ be a time shift operator. From the knowledge of Kalman filtering,~$\theta_{k+1,i} = T \theta_{k,i}$
		is ergodic. Furthermore,~$ T $ is measure-preserving when the Kalman filter enters steady state. It follows that~$ \theta_{k+1} = T \theta_k $ is ergodic and measure-preserving. The state transition from~$ \theta_k $ to~$ \theta_k^* $ can be seen as selecting the arrangement of different components in~$\theta_k$ with a fixed distribution~$ \mathrm{Pr}(H_i) = \alpha_i $. Therefore,~$ \theta_{k+1}^* = T \theta_{k}^* $ is ergodic and measure-preserving. From  {Proposition~\ref{proposition: correct converge of algorithm}},~$ \hat{\alpha}_i^{t} \rightarrow \alpha_i $ if the iteration times and the data collected are sufficient.  
	\end{pf*}
\end{theorem}
So far, the correct convergence of~$ \hat{\alpha}^{t} $ has been shown, which also means the fusion estimation with estimated probability weight will converge to~$ \bar{x}_k^B $. It should be noted that such an algorithm in the dynamic system will cause a shortage of computational resources due to the growth of the dataset~$ \Xi_L $. The addition of a newly received sample requires algorithm initialization and repeating iterations \eqref{eqn: E step}--\eqref{eqn: M step} again. 
To alleviate this defect, we use a stepwise EM (sEM) mode \cite{liang2009online}, as shown in Algorithm~\ref{Algorithm: sEM}.  We divide the received data over an infinite time horizon into multiple disjoint parts as~$ \bigcup_{s=1}^\infty \Xi^s = \Xi_\infty $. In this situation, the data set~$ \Xi^s $ is composed of
\begin{align*}
	\Xi^s = \left\{ \theta_{i_{s-1}+1}^*,\ldots , \theta_{i_{s}}^* \right\}, \quad 0 = i_0 < \ldots < i_s.
\end{align*}
The probability weights~$\alpha_i$ are then replaced by their corresponding estimated values~$\hat{\alpha}_i$.

\begin{algorithm}[htbp] 
	\caption{Fusion estimation with sEM algorithm} 
	\label{Algorithm: sEM}
	\hspace*{0.02in} {\bf Input:}~ 
	$ A, Q, C_i, R_i $
	\begin{algorithmic}[1] 
		\State Initialize~$ \mathcal{P}, W, \Psi, \Gamma_i $
		\For {\textbf{each}~$ k $} 
		\State Generate~$ \hat{x}_{k,i}$ by \eqref{eqn:standard x-}--\eqref{eqn:standard P}
		\State Calculate~$ \theta_{k}^* $ from~$ \hat{x}_{k,i}$
		\If {there exists a full batch~$\Xi^s$}  
		\State Generate~$\hat{\alpha}^{\infty} (\Xi^s) $ by~\eqref{eqn:fusion EM 1} and~\eqref{eqn:fusion EM 2} 
		\State $ \hat{\alpha} \leftarrow \eta_s \hat{\alpha}^{\infty} +  ( 1 - \eta_s) \hat{\alpha} $
		\State $ s \leftarrow s + 1 $
		\Else 
		\State Load~$ \theta_{k}^* $ into~$ \Xi^s $
		\EndIf
		\State $ \alpha \leftarrow \hat{\alpha} $, and update each~$ \alpha^*_{k,i} $ by \eqref{eqn: posterior probability} 
		\State Generate~$ \bar{x}_k^{B} $ by \eqref{eqn: Bayesian estimation}
		\EndFor
		\State \Return
	\end{algorithmic}
\end{algorithm}

The basic idea of this algorithm is to divide a large amount of online data into smaller datasets, and update~$ \hat{\alpha}^{t} $ in a stepwise pattern, where~$ \eta_s = 1/s $ is the step length for the moving average.
{We note that while the standard EM algorithm is guaranteed to converge, the stepwise variant introduces additional considerations. In particular, if the dataset is too small relative to the parameter set~$\Xi^s$, the estimates~$\hat{\alpha}^\infty$ may fluctuate significantly across iterations, potentially hindering convergence. As shown in Proposition~\ref{proposition: correct converge of algorithm}, this issue can be mitigated by ensuring~$\Xi^s$ is sufficiently large, albeit at increased computational cost. In practice, careful selection of step size and data segmentation is essential, especially when working with limited data.}

\section{Simulation} \label{section:simulation}
{The system used in the simulation is a linearized model of an aircraft, borrowed from~\cite{linehan19964}. The components of the system state~$x_k$ are horizontal velocity, vertical velocity, pitch rate, and pitch angle, respectively. The system matrix and output matrix are
\begin{align*}
	A = \left[ \begin{matrix}
		~~0.992 & ~~0.030 & -0.003 & -0.977 \\
		~~0.025 & ~~0.684 & ~~1.847 & -0.041 \\
		~~0.054 & -0.100 & ~~0.381 & -0.025 \\
		~~0.003 & -0.006 & ~~0.068 & ~~0.999 \\
	\end{matrix} \right].
\end{align*}
The measurement matrices of the sensors are
\begin{align*}
	C_1 &= \left[ \begin{matrix}
		~1 & ~~0 & ~~0 & ~~0
	\end{matrix} \right], 
	&C_2 &= \left[ \begin{matrix}
		~1 & ~~1 & ~~0 & ~~0
	\end{matrix} \right],
	\\
	C_3 &= \left[ \begin{matrix}
		~0 & ~~0 & ~~1 & -1
	\end{matrix} \right],
	&C_0 &= \left[ \begin{matrix}
		~0 & ~~1 & ~~0 & -1
	\end{matrix} \right].
\end{align*}
The covariances of the system noise~$Q$ and measurement noise~$R$ are identity matrices.} 
\begin{table}[htbp] 
	\center
	\caption{The parameters related to data packet shuffling}  \label{tab:event H}
	\renewcommand{\arraystretch}{1.5}
	\setlength{\tabcolsep}{1.2mm}{
	\begin{tabular}{ccccccc}
	\cline{1-7}
	\multicolumn{1}{c|}{$ H_i $}  				& $H_1$ & $H_2$ & $H_3$ & $H_4$ & $H_5$ & $H_6$ \\ \cline{1-7}
	\multicolumn{1}{c|}{$\mathrm{Pr}(H_i)$}  	& {\text{0.067}}& {\text{0.267}}& {\text{0.067}}& {\text{0.267}}& {\text{0.067}}& {\text{0.267}} \\ 
	\multicolumn{1}{c|}{$g_i$} 		& $ \left( \substack{ 1~2~3\\ 1~2~3 } \right) $ & $ \left( \substack{ 1~2~3\\ 3~1~2 } \right) $ & $ \left( \substack{ 1~2~3\\ 2~3~1 } \right) $ & $ \left( \substack{ 1~2~3\\ 2~1~3 } \right) $ & $ \left( \substack{ 1~2~3\\ 1~3~2 } \right) $ & $ \left( \substack{ 1~2~3\\ 3~2~1 } \right) $ 
	\\ [0.2cm]
	\hline
	\end{tabular}
	}
\end{table}
The sensor~$ i = 0 $ can receive their state estimation results from all other sensors. Table \ref{tab:event H} describes the specific permutations~$ g_i $ with corresponding probabilities~$\mathrm{Pr}(H_i)$.

\begin{figure}[t]
	\centering
	\includegraphics[width=8.8cm]{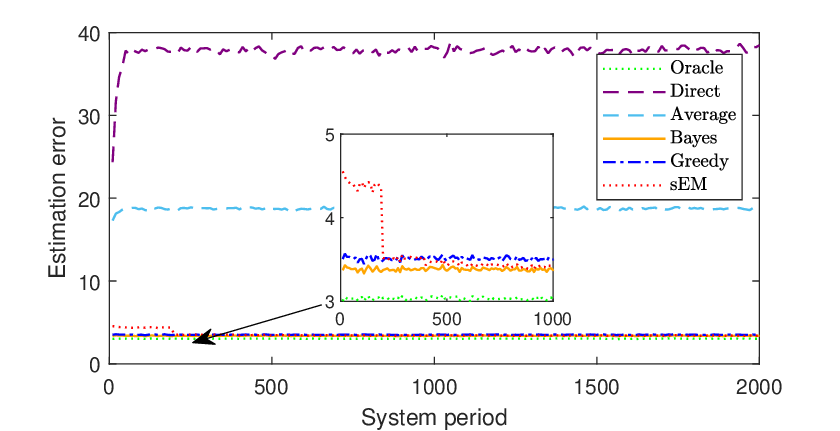}
	\caption{The fusion estimation errors in different algorithms.} \label{fig:simulation_estimation_error}
\end{figure}

\begin{figure}[t]
	\centering
	\includegraphics[width=8.8cm]{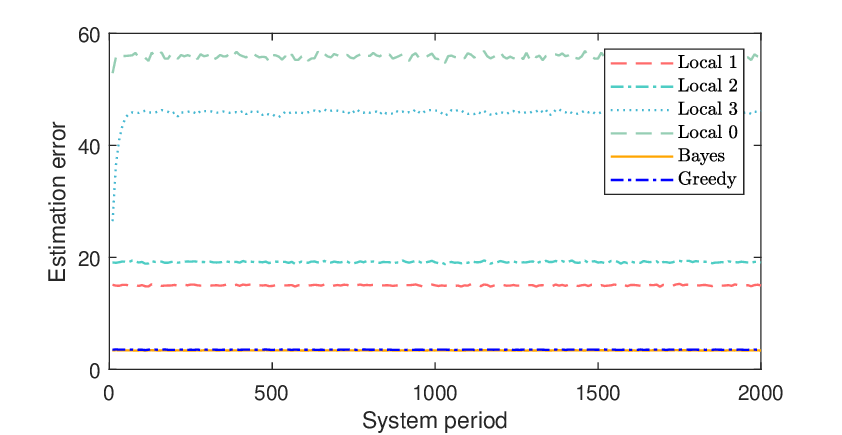}
	\caption{{The fusion estimation errors compared to local estimations.}}\label{fig:simulation_estimation_error_local}
\end{figure}

The state estimation errors of the fusion center in different scenarios are illustrated in Figure \ref{fig:simulation_estimation_error}. To demonstrate the performance improvements brought by the algorithms, we included three scenarios where fusion algorithms are not specifically designed. These scenarios are the ``Oracle'' curve, which represents the state estimation performance when the arrangement of data packets is correct; the ``Direct'' curve, which corresponds to the data packet arrangement in probability directly fused by the fusion center; and the ``Average'' curve, which computes the fused estimate as the average~$\frac{1}{N+1}\sum_{i=0}^{N}{\hat{x}_{k,i}}$ of all received data packets---a method whose result remains unchanged even if the packet sequence is disrupted. The curves ``Bayes,'' ``Greedy,'' and ``sEM'' correspond to the state estimation algorithms in Section~\ref{section:A Bayesian Approach}, Section~\ref{section:A Greedy Approach}, and Section~\ref{section:Learning Prior Information with the EM Algorithm}, respectively. In this simulation, the Bayesian approach achieves the best performance in state estimation. In addition, in Figure~\ref{fig:simulation_estimation_error_local}, we run the local estimators and compare them with the ``Bayes'' and ``Greedy'' ones in Figure~\ref{fig:simulation_estimation_error}. The results show that the proposed algorithm can significantly reduce the fusion estimation error compared with the local estimation.

\begin{table}[htbp] 
    \centering
    \caption{The frequency of events~$\theta_k^* \in \Omega_i$ matching~$H_i$.}  
    \label{tab:P Omega H}
    \renewcommand{\arraystretch}{1.5}
    \setlength{\tabcolsep}{0.8mm}{
    \begin{tabular}{ccccccc}
        \toprule[1pt]
        Events & $H_1$ & $H_2$ & $H_3$ & $H_4$ & $H_5$ & $H_6$ \\
        \midrule~
        $\Omega_i$\&$H_i$ & \textbf{{0.0557}} & \textbf{{0.2490}} & \textbf{{0.0560}} & \textbf{{0.2516}} & \textbf{{0.0571}} & \textbf{{0.2518}} \\
        \bottomrule[1pt]
    \end{tabular}
    }
\end{table}
The Greedy approach significantly improves compared to other non-specifically designed fusion estimation algorithms, indicating that it frequently achieves the correct data arrangement. We present the frequency of events~$\theta_k^* \in \Omega_i$ matching~$H_i$ occurring during the simulation process in Table \ref{tab:P Omega H}. The sum of the line in the table is~$ 92.12\% $ representing the accuracy percentage. 
For the situation without~$\mathrm{Pr}(H_i)$, the Bayesian approach equipped with the sEM process can gradually converge to the estimation performance with known~$\mathrm{Pr}(H_i)$, as shown in Figure \ref{fig:simulation_estimation_error}. 

\section{Conclusion} \label{Section:Conclusion}

This paper addresses the problem of fusion estimation in a multi-sensor system with potential data packets with disrupted identity issues by designing special fusion estimation algorithms. The mathematical formulation of data packet random permutation is presented by using a symmetry group. We develop two fusion estimation algorithms to improve the fusion process, a Bayesian approach and a greedy approach, both of which demonstrate expectation error-bounded performance when the \textit{a priori} probabilities of random permutations are known. The Bayesian approach uses posterior arrangement probabilities for weighted fusion, while the greedy approach effectively improves estimation performance by the most likely data arrangement. Furthermore, when the \textit{a priori} probabilities are unknown, we introduce the EM algorithm to estimate these probabilities and show its correct convergence. This provides a solution for real-world applications where complete knowledge cannot be available. 
The numerical simulations demonstrate the results.

\bibliographystyle{IEEETran}
\bibliography{references}

\end{document}